\documentclass[12pt, draftclsnofoot, onecolumn]{IEEEtran}

\usepackage{amsmath}              
\usepackage{amsthm}
\usepackage{amssymb}
\usepackage{amsfonts}
\usepackage{graphicx}
\usepackage{epstopdf}

\usepackage{setspace}

\newcommand{\mycase}[1]{\left\{#1\right.}
\newcommand{\myround}[1]{\left(#1\right)}

\newcommand{\MYfooter}{\smash{
\hfil\parbox[t][\height][t]{\textwidth}{}\hfil\hbox{}}}
\makeatletter
\def\@oddhead{\mbox{}$\date{\today}$ \rightmark \hfil }%
\def\@oddfoot{\MYfooter}%
\def\ps@IEEEtitlepagestyle{
  \def\@oddfoot{\mycopyrightnotice}
  \def\@evenfoot{}
}
\def\mycopyrightnotice{
  {\footnotesize
  \begin{minipage}{\textwidth}
  \flushleft
  978-1-5386-6378-3/18/\$31.00~\copyright2020 IEEE
  \end{minipage}
  }
}
\newtheorem{proposition}{Proposition}

\begin{document}\setstretch{1.5}

\title{Broadcast Approach for the Information Bottleneck Channel\thanks{The work of S. Shamai has been supported by the European Union's Horizon 2020 Research And Innovation Programme, grant agreement no. 694630.}}

\author{Avi Steiner and Shlomo Shamai (Shitz)\\
		\textit{Technion---IIT, Haifa 3200003, Israel \\
		steiner.avi@gmail.com, sshlomo@ee.technion.ac.il}
}

\maketitle

\begin{center}
	\today
\end{center}

\begin{abstract}
This work considers a layered coding approach for efficient transmission of data over a wireless block fading channel without transmitter channel state information (CSI), which is connected to a limited capacity reliable link, known as the bottleneck channel. Two main approaches are considered, the first is an oblivious approach, where the sampled noisy observations are compressed and transmitted over the bottleneck channel without having any knowledge of the original information codebook. The second approach is a non-oblivious decode-forward (DF) relay where the sampled noisy data is decoded, and whatever is successfully decoded is reliably transmitted over the bottleneck channel. The bottleneck channel from relay to destination has a fixed capacity $C$. We examine also the case where the channel capacity can dynamically change due to variable loads on the backhaul link. The broadcast approach is analyzed for cases that only the relay knows the available capacity for next block, and for the case that neither source nor relay know the capacity per block, only its capacity distribution. Fortunately, it is possible to analytically describe in closed form expressions, the optimal continuous layering power distribution which maximizes the average achievable rate. Numerical results demonstrate the achievable broadcasting rates.  
\end{abstract}

\section{Introduction and Preliminaries}

Block fading channel model is commonly used for wireless communications,
dominating the cases when mobile endpoints move slow relatively to the block coherence time.
In slowly varying fading channels the fading realization is fixed throughout a transmission block, giving rise to the block fading notion. By this model, the receiver can easily learn the channel characterization over the block,
thus we can assume perfect Channel State Information (CSI) only at the receiver side.
In most practical cases, there is no feedback channel to the transmitter, resulting
in its total unawareness of the instantaneous channel, yet it knows the channel statistics.

Consider the problem of transmitting over a block fading channel to a relay node,
which has to forward the received signal to a destination over a reliable link with a fixed capacity $C$, see Figure \ref{fig:model} for the schematic channel model. For Gaussian channels this is known as the bottleneck channel \cite{Tishby2004}. An overview of bottleneck problems with theoretical and practical application is presented in \cite{Zaidi20}. This channel model is also applicable for the evolving next generation 5G/6G cellular networks, where the communication with the promising architecture of the Cloud radio access network (C-RAN) introduces stringent requirements on the fronthaul capacity and latency \cite{5G_7456186}, \cite{5G_8753420}, and many other timely contributions. 

When transmitting over a block fading channel with receive CSI only, a broadcast approach may be considered on the transmission to maximize the average achievable rate. The broadcast approach, which is essentially a variable-to-fixed channel coding \cite{Verdu10variable-ratechannel}, was studied in \cite{ShitzSteiner03} for the MIMO fading channel with receiver CSI only. A finite capacity link to base-station subject to random fluctuations was studied in \cite{ROY_6613623} for the case of two users connecting to the same base-station. Another 
related overview of matrix monotonic optimization is studied in \cite{DBLP:journals/corr/abs-1810-11244}. Broadcast methods for the diamond channel, which is the two parallel relays channel, were studied in \cite{6642081}, \cite{8811604}. Broadcast approach for bottleneck channel with a known static bottleneck capacity channel is studied in \cite{SteinerShamai2019} and under bottleneck capacity uncertainty \cite{SteinerShamai2020}. 

In the classical Gaussian bottleneck problem, depicted in Figure \ref{fig:model}, define the random variable triplet $X-Y-Z$ forming a Markov chain, and related according to 
\begin{equation}\label{eqChannelModel}
Y = h\cdot X + N,
\end{equation}
where $X$ and $N$ are independent random variables,
with $N\sim N(0,1)$ being real Gaussian with a unit variance, $h$ is the channel fading fixed per codeword, and the fading gain $s=|h|^2$ with a signal to noise ratio (SNR) is $SNR=P\cdot s$, where the gain $s=1$ for a non-fading Gaussian channel, and $P$ is the transmission power $E[X^2]=P$. In the broadcast approach \cite{ShitzSteiner03}, discussed later, the channel model includes fading, where the fading gain $s$ has a known distribution. The bottleneck channel output $Z$ is a compressed version of $Y$ adhering to a limited capacity of the bottleneck channel $C$. It is of interest to maximize
\begin{equation}\label{maxEq}
\max\limits_{P(X), P(Z|Y)~s.t. I(Y;Z)\leq C} ~ I(X;Z)
\end{equation}
Evidently if $X$ is Gaussian it is well known by Tishby et al \cite{Tishby2004}, and \cite{tishby99information},
then also $Y-Z$ is a Gaussian channel, and the maximization result of \eqref{maxEq}
\begin{equation}\label{maxCompressionRate}
C_{Obliv} = I(X;Z) = \frac{1}{2}\log(1+P|h|^2) - \frac{1}{2}\log(1+P|h|^2\cdot \exp(-2C)),
\end{equation}
which follows immediately from the rate distortion approach,
that is the relay output can be represented by quantization of its input $Y$, 
\begin{equation}\label{eqRD}
 Z=Y+M,
\end{equation}
where the variance of $Y$ is $P|h|^2+1$ as determined by the channel model \eqref{eqChannelModel}, and the variance of $M$ representing the quantization noise, which
is determined by requiring $I(Z;Y) = C$, that is
\begin{equation}\label{eqVar}
E[M^2] = \frac{P|h|^2+1}{exp(2C)-1}
\end{equation}
The bottleneck gives reliable information rate
that can be transmitted from $X$ to $Z$, when the relay $Y$
operates in an oblivious way (it has no knowledge about
the codebook and can not decode the message). For a non-oblivious decode-forward (DF) approach the result is immediate, as the relay may decode the data, and then transmit over the limited bandwidth channel $Y-Z$ at rate $C$. Therefore the achievable transmission rate is the minimum of the two channels capacity, 
\begin{equation}
C_{DF}=\min \{ \frac{1}{2}log(1+P|h|^2),C \}.
\end{equation}
Another common setting in cellular uplink is a variable availability of capacity on the backhaul. This may be the result of variable loads on the network over time. Traffic congestion of internet data may lead to changing availability levels of the backhaul \cite{ROY_6613623}. On the bottleneck channel this means that the relay-destination link capacity $C$ is a random variable. It may be assumed that the transmitter is aware of the average capacity, and its distribution, however like in case of the wireless fading channel, the capacity variability dynamics may not allow the time needed for a useful feedback to the transmitter. When relay is fully aware of the current bottleneck currently available capacity for the received codeword, it can match the transmission data rate. However, when relay is not aware of the available capacity per codeword, it has to perform successive refinement source coding \cite{Tian08} matched to the capacity distribution. This problem is analyzed in Section \ref{UncertaintySec}.

\subsection{Broadcast Approach Preliminaries}\label{sub_broadcast}
Consider a transmitted signal $X$ composed of multi-layer coded information, in a continuum of layers, such that each code layer receives an infinitesimal power $\rho(u)du$. The broadcast approach was introduced in detail in \cite{ShitzSteiner03}. We
briefly review the principles of the broadcast approach. Consider the channel model in (\ref{eqChannelModel}), where the channel gain $u$ is a block fading random variable, known at the receiver only. Transmitter is only aware of its distribution, but not its realization per codeword. The incremental rate as function of power allocation, for a Gaussian fading channel, is \cite{ShitzSteiner03}
\begin{eqnarray}\label{S_4}
dR(u) = \frac{1}{2}\cdot \frac{\rho(u)udu}{1+I(u)u}
\end{eqnarray}
where $I(u)$ is the residual interference function, such that
$I(0)=P$, and $\rho(u)=-I'(u)$ is the power allocation density
function. The total allocated rate as function of $s$ is thus
\begin{eqnarray}\label{S_4_1}
R(s) = \frac{1}{2}\cdot \int\limits_0^s\frac{\rho(u)udu}{1+I(u)u}
\end{eqnarray}
The maximal average rate is expressed as follows
\begin{eqnarray}\label{S_5}
R_{bs,avg} = \max\limits_{I(u)} \frac{1}{2} \int\limits_0^{\infty} du
(1-F_s(u)) \frac{\rho(u)u}{1+I(u)u}
\end{eqnarray}
where $F_s(u)$ is the cumulative distribution function (CDF) of
the fading gain random variable. It may be shown, \cite{ShitzSteiner03},
that the optimal power allocation is given by
\begin{eqnarray}\label{S_7}
I_{opt}(u)=\mycase{
	\begin{array}{ll}
	P & u< u_0\\
	\frac{1-F_s(u)-u\cdot f_s(u)}{u^2f_s(u)} & u_0\leq u\leq u_1\\
	0 & u> u_1
	\end{array}}
\end{eqnarray}
where $u_0$ and $u_1$ are obtained from the boundary conditions
$I_{opt}(u_0)=P$, and $I_{opt}(u_1)=0$, respectively.

Interestingly, the optimal allocated rate can be expressed in closed form by substituting the optimal power allocation (\ref{S_7}) into the cumulative layering rate in (\ref{S_4_1}), by
\begin{eqnarray}\label{S_4_2}
R_{opt}(s) =\mycase{
	\begin{array}{ll}
	0 & s< u_0\\
	\log(s/u_0)+\frac{1}{2} \log\left( \frac{f_s(s)}{f_s(u_0)} \right) & u_0\leq s\leq u_1\\
	\log(u_1/u_0)+\frac{1}{2} \log\left( \frac{f_s(u_1)}{f_s(u_0)} \right) & s > u_1
	\end{array}}
\end{eqnarray}

\begin{figure}[]
	\centering
	\includegraphics[width=0.8\textwidth]{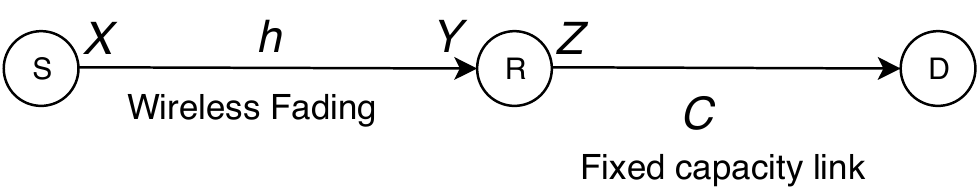}
	\caption{Information bottleneck fading channel system model block diagram.}
	\label{fig:model}
\end{figure}

\section{Broadcast Approach for Fading Information Bottleneck Channel}

Consider a fading channel on the wireless link to $Y$, where $s=|h|^2$ is the block fading gain with a unit variance. Under a slowly fading channel, the random variable gain $s$ changes independently from codeword to codeword, and remains fixed over the codeword. The channel model for $Z$ can be expressed by its block fading gain, under an oblivious approach
\begin{equation}
Z = sqrt(FPR_{eq}) X + N,
\end{equation}
where $N$ is a unit variance Gaussian noise, and the equivalent fading power gain is
\begin{equation}\label{snreq}
FPR_{eq} = \frac{s(1-\exp(-2C))}{1+s\cdot P\cdot \exp(-2C)}
\end{equation}
which is directly obtained from \eqref{eqVar}. It may be observed that $FPR_{eq}$ is finite for $s\geq 0$, and at the limit of $s\rightarrow\infty$ becomes
\begin{equation}\label{snreq_limit}
\lim\limits_{s\rightarrow\infty} FPR_{eq} = (\exp(2C)-1)/P
\end{equation}

\subsection{Oblivious Bottleneck Channel Approach}
On the oblivious approach, the received codeword in $Y$ is not decoded, and there is no information of its codebook, therefore, the compression of $Y$ into $Z$ is performed accounting for the distribution of $Y$ only. Under a fading channel model, the ergodic capacity of the bottleneck fading channel is determined by 
\begin{eqnarray}\label{erg_Obliv}
C_{Obliv,Erg} &=& E_s[\frac{1}{2}\log(1+P\cdot FPR_{eq})] \nonumber \\
 &=&   E_s\left[ \frac{1}{2} \log\left( 1+ \frac{s\cdot P\cdot (1-\exp(-2C))}{1+s\cdot P\cdot \exp(-2C)} \right) \right] \label{eqn111}
\end{eqnarray}
\subsubsection{Single Layer Coding}
Using a single layer coding approach for the fading channel, the achievable average rate depends on the allocated rate, and the fading distribution in the following way. In the oblivious communication scheme, let the transmitter allocate a rate $R_{1,obliv}$ as function of a fading threshold parameter $s_{th}$. Then, the decoding of the noisy compressed signal $Z$ is possible for fading gains  $s\geq s_{th}$, and the allocated rate corresponds to 
\begin{equation}\label{eq1}
R_{1,obliv} = \frac{1}{2}log(1+FPR_{eq}(s_{th})P)  = \frac{1}{2}log\left(  1+ \frac{s_{th}(1-\exp(-2C))P}{1+s_{th}\cdot P\cdot \exp(-2C)}\right)
\end{equation}
Since $FPR_{eq}(s_{th})$ is a monotonic function of $s_{th}$, the rate $R_{1,obliv}$ can be achieved for any fading gain $s\geq s_{th}$, and therefore the average rate with a single layer is 
\begin{equation}\label{eq2}
R_{1,obliv,avg} = (1-F_s(s_{th})) \frac{1}{2}log\left(  1+ \frac{s_{th}(1-\exp(-2C))P}{1+s_{th}\cdot P\cdot \exp(-2C)}\right)
\end{equation}
where $F_s(x)$ is the CDF of the channel fading gain $s$, and the outage capacity is 
\begin{equation}\label{eq3}
R_{1,obliv,avg} = \max\limits_{s_{th}\geq 0} ~ (1-F_s(s_{th})) \frac{1}{2}log\left(  1+ \frac{s_{th}(1-\exp(-2C))P}{1+s_{th}\cdot P\cdot \exp(-2C)}\right)
\end{equation}
\subsubsection{Continuous Broadcast Approach}

We derive here the continuous broadcasting approach, where the transmitted signal $X$ is multi-layer coded, in a continuum of layers, such that each code layer receives an infinitesimal power $\rho(u)du$. The broadcast approach
was introduced in detail in \cite{ShitzSteiner03}. The channel model here can be expressed in its equivalent form, where the fading gain is an equivalent fading gain $\nu=FPR_{eq}$ from (\ref{snreq}), which also depends on the channel fading gain $s$ distribution, and on the bottleneck channel capacity, as well as the transmission power. In this oblivious bottleneck channel, the broadcast approach is optimized for a fading distribution $F_\nu(u)$ of (\ref{snreq}). The derivation of the optimal power distribution is then directly derived as described in Section \ref{sub_broadcast}. Clearly for high bottleneck channel capacity $C\rightarrow\infty$, then $FPR_{eq}\rightarrow s$.

\subsection{A Non-Oblivious Decode Forward Bottleneck Channel Approach}
On the non-oblivious DF approach, the received codeword in $Y$ can be decoded, and then all the decoded data up to rate $C$ can be reliably conveyed to $Z$. Under a fading channel model, the non-oblivious ergodic capacity of the bottleneck fading channel $C_{DF,Erg}$ provides an ergodic upper bound which is not achievable under a block fading model, where each codeword is transmitted over a relatively short duration with a single channel realization, without capturing the full distribution of $s$. For the slowly block fading channel it is beneficial to transmit a multi-layered codeword when transmitter has no channel state information (CSI). Under this model, non-oblivious ergodic capacity of the bottleneck fading channel is formulated as
\begin{equation}\label{eq_erg}
C_{DF,Erg} = E_s[\min\{\frac{1}{2}log(1+sP), C\}].
\end{equation}
which corresponds to a block fading model, where the transmission and decoding are done over a single fading realization, due to a slow fading nature of the channel.
\subsubsection{Single Layer Coding}
Using a single layer coding approach for the fading channel, the achievable average rate depends on the allocated rate, the bottleneck channel capacity $C$ and the fading distribution in the following way. In the non-oblivious communication scheme, let the transmitter allocate a rate $R_{1,DF}$ as function of a fading threshold parameter $s_{th}$. Then, the decoding of the noisy compressed signal $Z$ is possible for fading gains  $s\geq s_{th}$, and the allocated rate corresponds to 
\begin{equation}\label{eq1_SL}
R_{1,DF} = \frac{1}{2}log(1+ s_{th}P)
\end{equation}
where the rate $R_{1,DF}$ is selected such that $R_{1,DF}\leq C$, and can be achieved for any fading gain $s\geq s_{th}$, and conveyed reliably over the bottleneck channel after decoding. Therefore the average rate with a single layer is 
\begin{equation}\label{eq2_SL}
R_{1,DF,avg}(s_{th}) = (1-F_s(s_{th}))\cdot R_{1,DF} = (1-F_s(s_{th}))\cdot \frac{1}{2}log(1+ s_{th}P) 
\end{equation}
and hence the outage capacity of the non-oblivious channel is given by
\begin{equation}\label{eq3_SL}
R_{1,DF,avg} = \max\limits_{s_{th}\geq 0} (1-F_s(s_{th}))\cdot \min( C, \frac{1}{2}log(1+ s_{th}P) )
\end{equation}

\subsubsection{Continuous Broadcast Approach}

We derive here the continuous broadcasting approach for the non-oblivious DF approach, where the transmitted signal $X$ is multi-layer coded, in a continuum of layers. The received signal $Y$ is decoded layer-by-layer in a successive decoding manner. All the successfully decoded layers with a total rate up to the bottleneck channel capacity $C$ can be reliably conveyed over the bottleneck channel. The broadcast approach optimization goal is to maximize the average transmitted rate over the bottleneck channel in this block fading channel model. We formulate here the optimization of power density distribution function $\rho_{opt}(u)$ so that average transmission rate is maximized under the bottleneck channel capacity constraint.

\begin{proposition}\label{thm:RbsNonObliv}
	For the non-oblivious block fading bottleneck channel, the total \emph{expected average achievable rate of the broadcast approach} is obtained by the following residual power distribution function
\begin{eqnarray}\label{prop_1}
I_{opt}(u) = \arg \max\limits_{I(u)} \frac{1}{2} \int\limits_0^{\infty} du
(1-F_s(u)) \frac{\rho(u)u}{1+I(u)u} , ~~ s.t. ~ \int\limits_0^{\infty} du \frac{\rho(u)u}{1+I(u)u} \leq C
\end{eqnarray}
where $F_s(u)$ is the CDF of the fading gain random variable, and $C$ is the bottleneck channel capacity. The optimal power allocation $I_{opt}(u)$ is given by
\begin{eqnarray}\label{prop_2}
I_{opt}(u)=\mycase{
	\begin{array}{ll}
	P & u< u_0\\
	\frac{1-F_s(u)+\lambda_{opt}-u\cdot f_s(u)}{u^2f_s(u)} & u_0\leq u\leq u_1\\
	0 & u> u_1
	\end{array}}
\end{eqnarray}
where $\lambda_{opt}\geq 0$ is a Lagrange multiplier specified by
\begin{eqnarray}\label{prop_3}
\lambda_{opt} = -u_1\cdot f_s(u_1) -1+F_s(u_1)
\end{eqnarray}
and for any $\lambda_{opt}>0$,  
\begin{eqnarray}\label{prop_4}
u_1^2 \cdot f_s(u_1) = \exp(2C)\cdot u_0^2 \cdot f_s(u_0)
\end{eqnarray}	
\end{proposition}

\begin{proof}
	The proof is based on solving a constrained minimization problem using variational calculus tool. See Appendix A, in Section \ref{Appendix}.
\end{proof}

\section{Bottleneck Capacity Uncertainty}\label{UncertaintySec}
A common case in cellular uplink is a variable availability of capacity on the backhaul. This may be the result of variable loads on the network over time. Traffic congestion of internet data may lead to changing availability levels of the backhaul \cite{ROY_6613623}. On the bottleneck channel this means that the relay-destination link capacity $C$ is a random variable. It may be assumed that the transmitter is aware of the average capacity, and its distribution, however like in case of the wireless fading channel, the capacity variability dynamics may not allow time for feedback to the transmitter. 
The following subsection considers the case that relay is fully aware of the current bottleneck available capacity for each received codeword.
 
Consider a bottleneck channel capacity discrete random variable $C_b$, which may admit to $N$ capacity values $\{C_i\}_{i=1}^{N}$, such that $C_1\leq C_2 \leq \cdots \leq C_N$ with corresponding probabilities $\{p_{b,i}\}_{i=1}^{N}$, such that $p_{b,i}\geq 0$ and $\sum_{i=1}^{N}p_{b,i}=1$. The average capacity of the bottleneck channel is 
\begin{eqnarray}\label{avg_capacity_1}
C_{avg} = \sum_{i=1}^{N}p_{b,i}C_i
\end{eqnarray}	

\subsection{Oblivious Bottleneck Channel Approach}
On the oblivious approach, the received codeword in $Y$ is not decoded, and there is no information of its codebook, therefore, the compression of $Y$ into $Z$ is performed accounting for the bottleneck channel capacity $C_i$, which is available to the relay, and the distribution of $Y$ only. Under a fading channel model, the ergodic capacity of the bottleneck fading channel under the uncertainty bottleneck capacity is determined by 
\begin{eqnarray}\label{erg_OblivMult}
C_{UC,Obliv,Erg} &=& E_{s,C_b}[\frac{1}{2}\log(1+P\cdot FPR_{eq}(Cb))] \nonumber \\
&=&  \frac{1}{2} \sum_{i=1}^{N}p_{b,i}\cdot E_s\left[  \log\left( 1+ \frac{s\cdot P\cdot (1-\exp(-2C_i))}{1+s\cdot P\cdot \exp(-2C_i)} \right) \right] \label{eqn111_mult}
\end{eqnarray}
which directly follows from (\ref{erg_Obliv}) and (\ref{avg_capacity_1}).  
\subsubsection{Single Layer Coding}
Using a single layer coding approach for the fading channel, the achievable average rate depends on the allocated rate, the bottleneck capacity and the fading distribution in the following way. Let the transmitter allocate a rate $R_{1,obliv}$ as function of a fading threshold parameter $s_{th}$, which is different from the fading threshold used in (\ref{eq1}), 
\begin{equation}\label{eq1Mult}
R_{1,obliv} = \frac{1}{2}log(1+s_{th}P).
\end{equation}
For a given availability of the bottleneck channel, denoted by capacity realization $C_i$, the required minimal equivalent FPR for successful decoding is $FPR_{eq}(C_i)\geq s_{th}$, i.e. decoding will succeed for fading gain $s\geq0$ such that
\begin{equation}\label{eq2Mult}
s_{th} \leq \frac{s\cdot(1-\exp(-2C_i))}{1+s\cdot P\cdot \exp(-2C_i)}
\end{equation}
where $C_i$ is the bottleneck capacity realization for a give codeword that was transmitted over a fading channel with a fading gain realization $s$. Therefore the average rate with a variable bottleneck capacity and a single layer transmission is 
\begin{equation}\label{eq3Mult}
R_{UC,1}(s_{th}) = \frac{1}{2} \sum_{i=1}^{N}p_{b,i}\cdot \left( 1-F_s\left(\frac{s_{th}}{1-\exp(-2C_i)(1+P\cdot s_{th})} \right) \right) log\left(  1+s_{th}\cdot P \right)
\end{equation}
where $F_s(x)$ is the fading gain CDF. The outage capacity is then 
\begin{equation}\label{eq4Mult}
R_{UC,1,Obliv,avg} = \max\limits_{s_{th}\geq 0} ~ R_{UC,1}(s_{th})
\end{equation}

\subsubsection{Continuous Broadcast Approach}

We derive here the continuous broadcasting approach, where the transmitted signal $X$ is multi-layer coded, in a continuum of layers, such that each code layer receives an infinitesimal power allocation corresponding to an equivalent fading gain parameter. Since the transmitter is not aware of the bottleneck capacity per codeword, but only its distribution, and average value, the following optimization flow is used for the continuous broadcast approach optimization. 

The combined equivalent channel viewed by the transmitter  
\begin{equation}\label{snreqCV}
FPR_{eq}(s,C_b) = \frac{s(1-\exp(-2C_b))}{1+s\cdot P\cdot \exp(-2C_b)}, ~~ s=|h|^2,
\end{equation}
Continuous broadcast approach is optimized for a fading distribution $F_\mu(u)$ where $\mu=FPR_{eq}(s,C_b)$ (\ref{snreqCV}) : equivalent channel gain depending on the fading gain realization $s$, and bottleneck channel capacity $C_b$ available per codeword. The cdf of this fading gain is
\begin{equation}\label{snreqvpdf}
F_\mu(u)=\sum_{i=1}^{N}p_{b,i} F_s \myround{\frac{u}{1-(1+Pu)\exp(-2C_i)}}
\end{equation}
The main result here is expressed on the following proposition
\begin{proposition}	
	The power distribution, which maximizes the
	expected rate over the oblivious bottleneck channel is
	\begin{eqnarray}\label{eq:general_Ix1}
	I(x) = \mycase{
		\begin{array}{cl}
		\frac{1-F_\mu(x)-x\cdot f_\mu(x)}{x^2f_\mu(x)} & ,~x_0\leq x\leq x_1 \\
		0 & ,~else
		\end{array}}
	\end{eqnarray}
	where $x_0$ is determined by $I(x_0)=P$, and $x_1$ by
	$I(x_1)=0$. And the broadcasting rate is expressed as function of the $FPR_{eq}$ distribution $F_\mu(u)$
	\begin{eqnarray}\label{eq5Mult_BS}
	R_{opt}(s) =\mycase{
		\begin{array}{ll}
		0 & s< x_0\\
		\log(s/x_0)+\frac{1}{2} \log\left( \frac{f_\mu(s)}{f_\mu(x_0)} \right) & x_0\leq s\leq x_1\\
		\log(x_1/x_0)+\frac{1}{2} \log\left( \frac{f_\mu(x_1)}{f_\mu(x_0)} \right) & s > x_1
		\end{array}}
	\end{eqnarray}
	
\end{proposition}	

\begin{proof}
	The proof is a direct derivation of the broadcast approach optimization \cite{ShitzSteiner03} for the power distribution under an equivalent channel model that includes the relayed signal after compression to a rate which matches the bottleneck channel capacity.
	The channel model for the relayed signal $Z$ can be expressed by its block fading gain, under an oblivious approach
	\begin{equation}
	Z = \sqrt{FPR_{eq}} \cdot X + N,
	\end{equation}
	where $N$ is a unit variance Gaussian noise, and $FPR_{eq}(s,C_b)$ is specified in (\ref{snreqCV}),
	which is directly obtained from the wireless channel model as stated in the introduction.
\end{proof}

\subsection{Unknown Bottleneck Channel Capacity at the Relay}
Consider the case that bottleneck channel capacity dynamics it too fast even for the relay to know the available capacity per codeword. The transmitter and relay know only the capacity distribution as specified in (\ref{avg_capacity_1}). In this case, the relay can employ successive refinement source coding \cite{successiveCover1991} on it input signal $Y$, which has a Gaussian distribution. On the oblivious relaying approach, successive refinement coding is to be performed as function of the capacity distribution function. For the discrete distribution $\{C_i\}_{i=1}^{N}$, the source coding layers code rate will be $R_{sc,1}=C_1$, $R_{sc_,2}=C_2-C_1$, ..., $R_{sc_,N}=C_N-C_{N-1}$. In successive refinement source coding, the refinement layer is encoded using the previous layers as side information. This means that decoding process also includes ordered decoding of the source layers, where target decoded layer is decoded using all previous layers as side information.

For a Gaussian channel it is well known that the Gaussian source is successively refinable \cite{successiveCover1991}, and even some more general sources \cite{Berger2001}. This means that maximization of the expected rate with oblivious relaying the exact same rate can be achieved as specified in (\ref{eq5Mult_BS}). This means that as long as the relay can perform successive refinement source coding matched to backhaul capacity distribution, it does not help and cannot increase expected achievable rate if the relay is informed with the available capacity per codeword.

\section{Numerical Results}
The following section provides some examples of achievable rates with single layer coding and continuous broadcasting, with comparison to the ergodic bound, for the block fading information-bottleneck channel. Both oblivious, and non-oblivious DF approaches are evaluated for the known fixed capacity bottleneck channel capacity. The numerical results are calculated for a Rayleigh fading channel, where $F_{\nu}(u) = 1-\exp(-u)$, with a static bottleneck capacity $C$.
\begin{figure}[]
	\centering
	\includegraphics[width=0.8\textwidth]{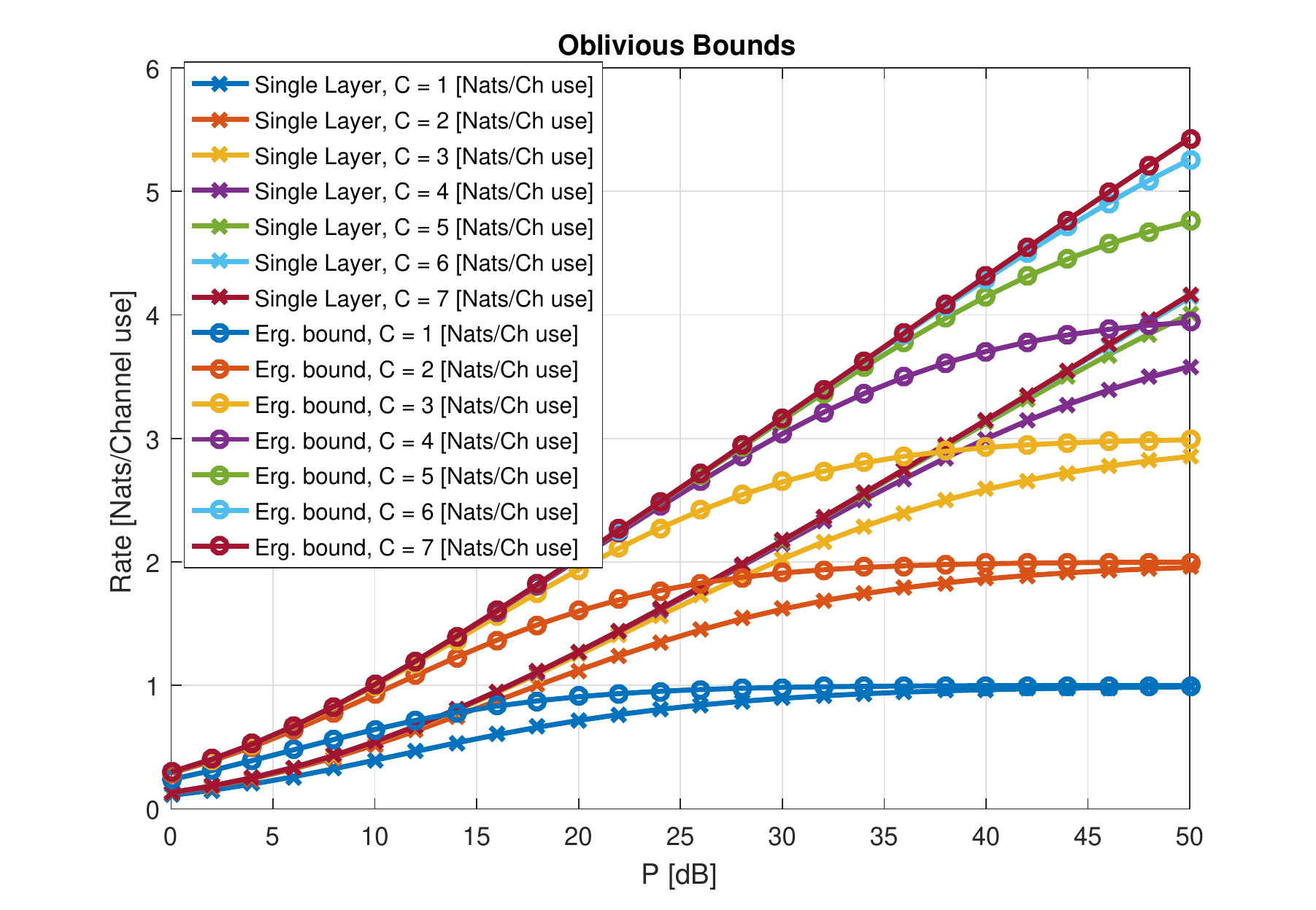}
	\caption{Oblivious approach of single layer coding vs. ergodic capacity, as function of bottleneck channel capacity.}
	\label{fig:Obliv_singlelayer}
\end{figure}
Figure \ref{fig:Obliv_singlelayer} demonstrates the achievable rate with a single layer coding oblivious approach, specified in (\ref{eq1}), compared to the ergodic oblivious bound, as specified in (\ref{erg_Obliv}). The multiple curves correspond to different values of the bottleneck channel capacity $C$. It is clear from the results here that for small $C$ the single layer asymptotically achieves the ergodic bound, while for $C\geq 3$ there is a large gap of the single layer approach to the ergodic bound, which may be narrowed down by using the broadcast approach.
\begin{figure}[]
	\centering
	\includegraphics[width=0.8\textwidth]{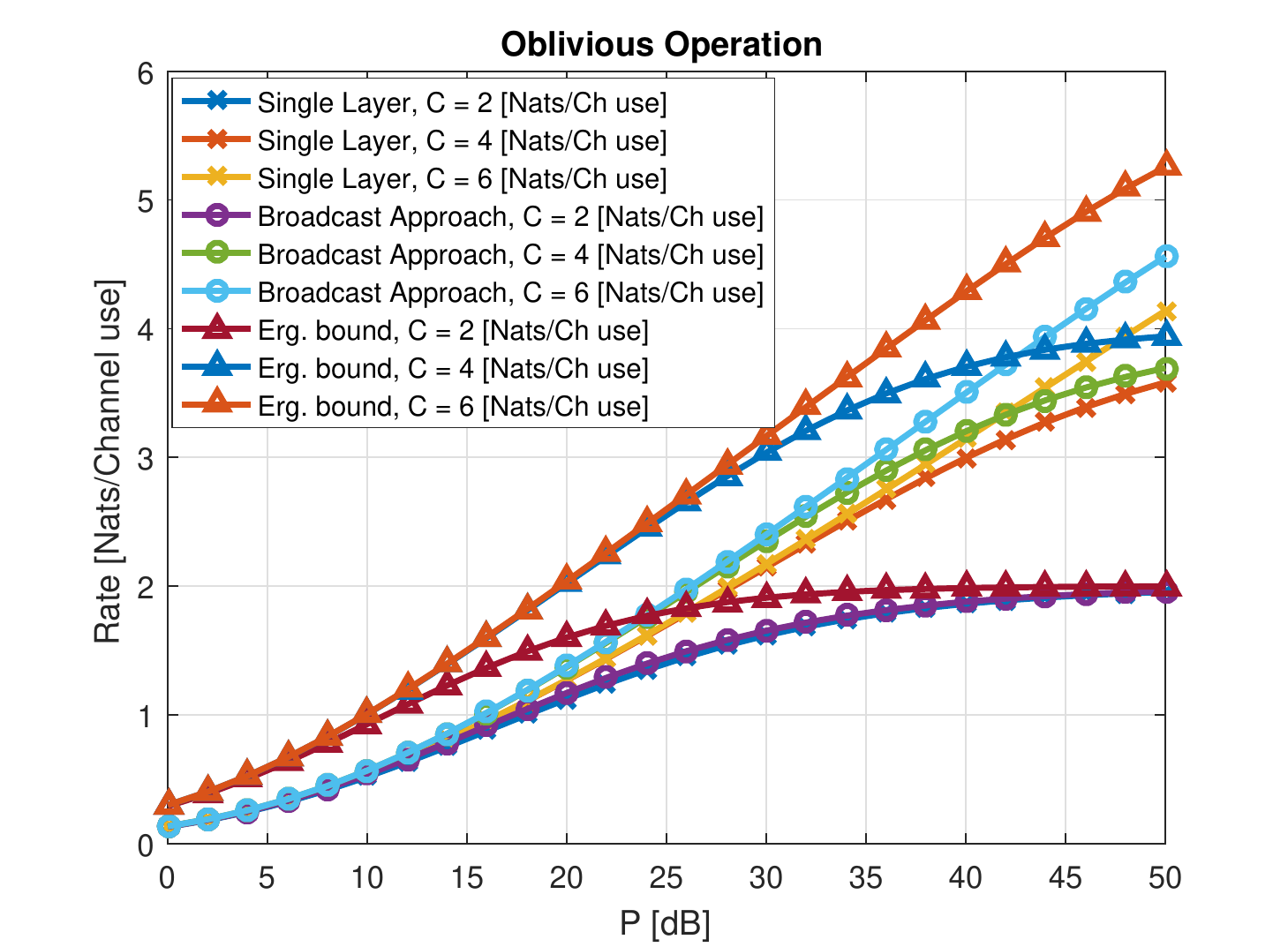}
	\caption{Oblivious approach of single layer coding and broadcast approach compared to the ergodic capacity, as function of bottleneck channel capacity.}
	\label{fig:Obliv}
\end{figure}
Figure \ref{fig:Obliv} demonstrates the achievable rates with a single layer coding oblivious approach, specified in (\ref{eq1}), compared to the oblivious broadcast approach, specified in (\ref{S_5})-(\ref{S_7}) where $F_\nu(u)$ is defined by $\nu=FPR_{eq}$ from (\ref{snreq}), and the ergodic oblivious bound (\ref{erg_Obliv}). The multiple curves correspond to different values of the bottleneck channel capacity $C$. It may be observed from the results that the higher the bottleneck channel capacity $C$, the higher is the oblivious broadcast approach gain compared to the single layer coding approach. 
\begin{figure}[]
	\centering
	\includegraphics[width=0.8\textwidth]{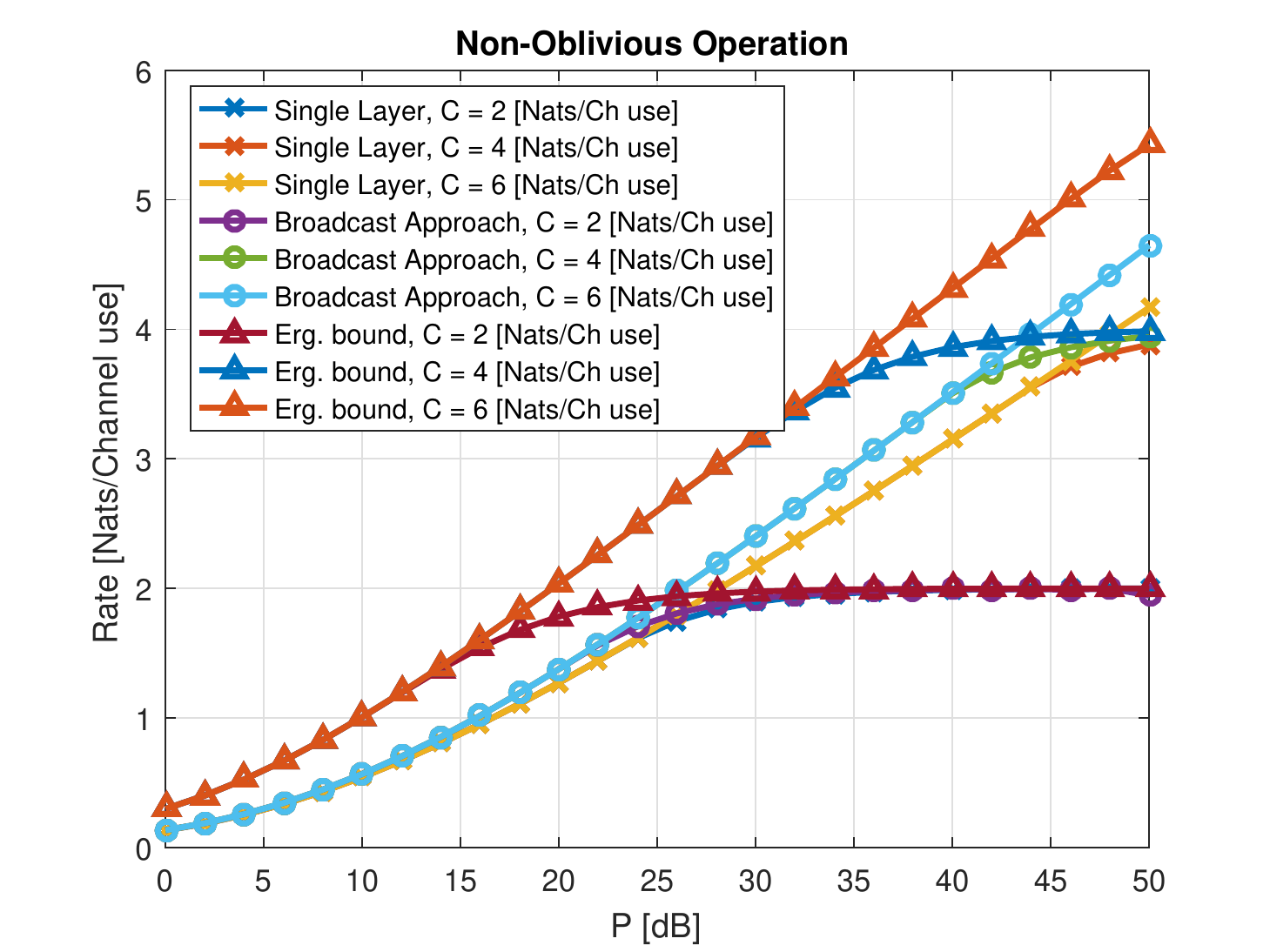}
	\caption{Non-Oblivious single layer coding and broadcast approach compared to the ergodic capacity, as function of bottleneck channel capacity.}
	\label{fig:nonObliv}
\end{figure}
Figure \ref{fig:nonObliv} demonstrates the achievable rates with a single layer coding non-oblivious approach, specified in (\ref{eq1_SL}), compared to the non-oblivious broadcast approach, specified in Proposition \ref{thm:RbsNonObliv}, and the ergodic non-oblivious bound (\ref{eq_erg}). The multiple curves correspond to different values of the bottleneck channel capacity $C$. It may be observed from the results that the higher the bottleneck channel capacity $C$, the higher is the non-oblivious broadcast approach gain compared to the single layer coding non-oblivious approach. 
\begin{figure}[]
	\centering
	\includegraphics[width=0.8\textwidth]{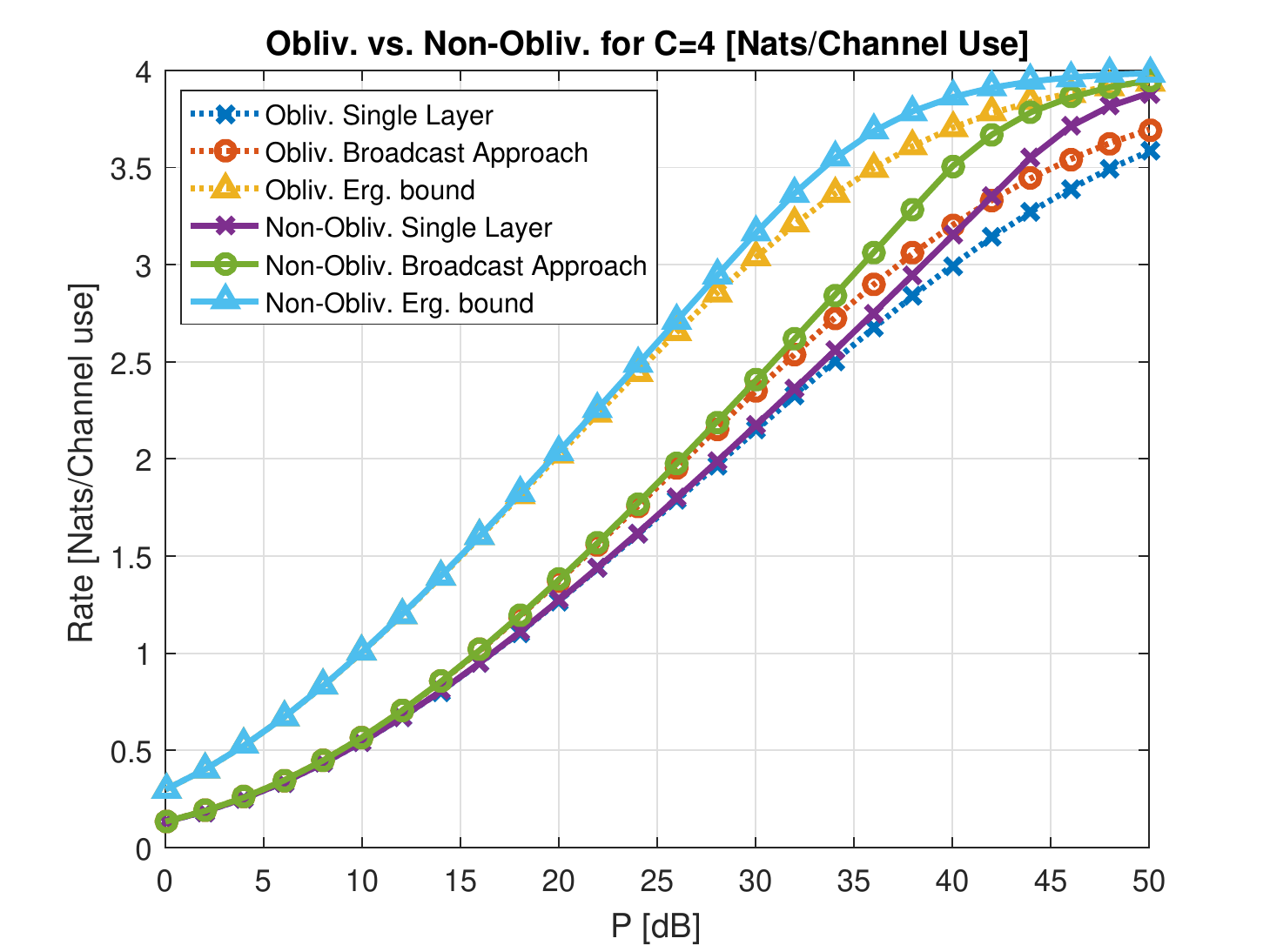}
	\caption{Oblivious vs. Non-Oblivious single layer coding and broadcast approach compared to the ergodic capacity, for bottleneck channel capacity of $C=4$ [Nats/Channel Use].}
	\label{fig:OblivVsNonObliv}
\end{figure}
Figure \ref{fig:OblivVsNonObliv} demonstrates the achievable rates with a non-oblivious approach as compared to an oblivious approach, for a bottleneck channel capacity $C=4~[Nats/Channel use]$. It can be observed here that at high SNR region the gain of the broadcast approach compared to single layer coding is higher with a non-oblivious approach.

\subsection{The Bottleneck Capacity Uncertainty}
In this subsection the impact of uncertainty in the bottleneck channel capacity is evaluated for a two state bottleneck capacity random variable, where ($C_1, C_2$) are the two possible capacity values, corresponding to probabilities ($p_1, (1-p_1)$). The comparison is done for the oblivious approach with $p_1=1/3$, and $C_{avg}=p_1C_1+(1-p_1)C_2$, and $C_{avg}=C$ on all cases, meaning that the average available capacity is equal to the deterministic capacity setting.
The numerical results are calculated for a Rayleigh fading channel, where $F_s(u) = 1-\exp(-u)$.

Figure \ref{fig:OblivFixedAvg_1} demonstrates the achievable rate with a single layer coding oblivious approach, specified in (\ref{eq1}) for a fixed bottleneck capacity $C$, compared to the uncertain bottleneck capacity specified in (\ref{eq4Mult}). The multiple curves correspond to different values of the fixed/average bottleneck channel capacity $C$.

Figure \ref{fig:OblivFixedAvg_1} demonstrates the achievable rate with the oblivious broadcast approach, specified in (\ref{S_5})-(\ref{S_7}) where $F_\nu(u)$ is defined by $\nu=FPR_{eq}$ from (\ref{snreq}) for a fixed bottleneck capacity $C$. This is compared to the uncertain bottleneck capacity specified in (\ref{eq5Mult_BS}). The multiple curves correspond to different values of the fixed/average bottleneck channel capacity $C$.

Figure \ref{fig:OblivFixedAvg_All} compares the single layer with the broadcast approach and ergodic capacity for $C=4$ Nats/Channel use, for a fixed bottleneck capacity and a two state bottleneck capacity. As may be noticed from the numerical results, the penalty of a random bottleneck is mainly on the high SNRs where the achievable average rate is in the range of the bottleneck capacity. 
 
\begin{figure}[h!]
	\centering
	\includegraphics[width=0.8\textwidth]{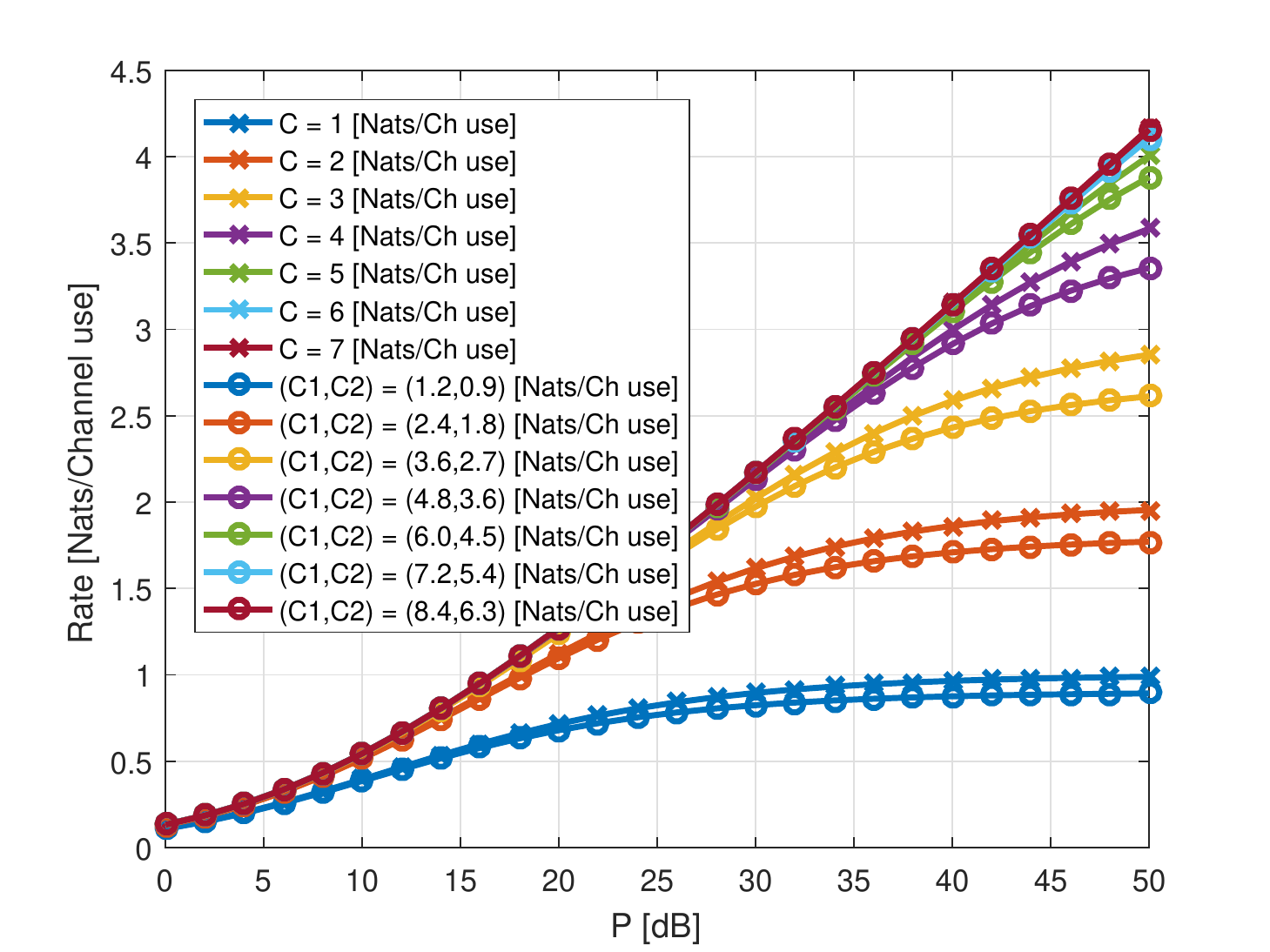}
	\caption{Oblivious single layer coding, comparison of fixed vs. average bottleneck channel capacity $C=\sum_{i=1}^{2}p_{b,i}C_i$, and $p_{b,1}=1/3$, $p_{b,2}=1-p_{b,1}$.}
	\label{fig:OblivFixedAvg_1}
\end{figure}
\begin{figure}[h!]
	\centering
	\includegraphics[width=0.8\textwidth]{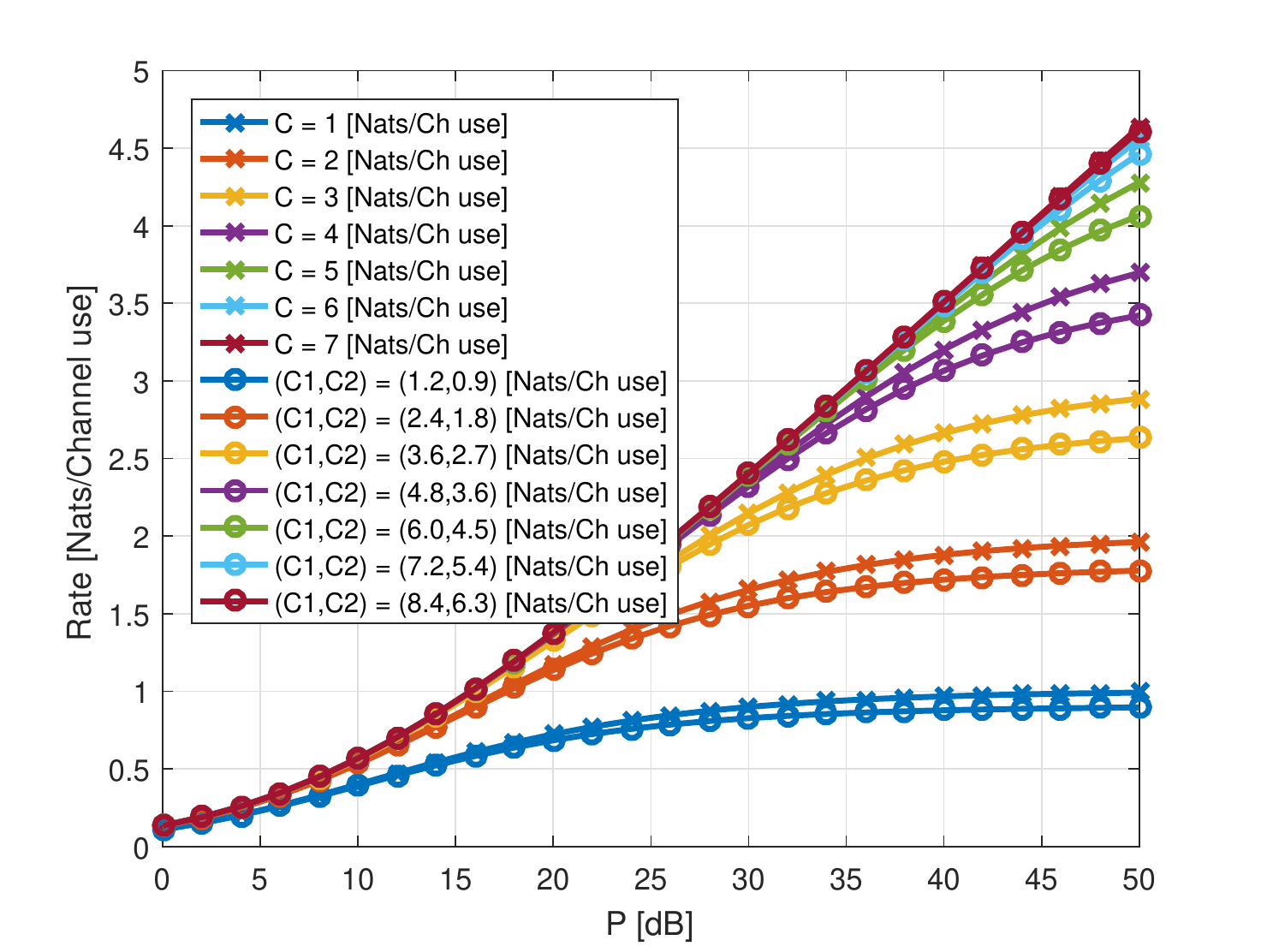}
	\caption{Oblivious Broadcast approach, comparison of fixed vs. average bottleneck channel capacity $C=\sum_{i=1}^{2}p_{b,i}C_i$, and $p_{b,1}=1/3$, $p_{b,2}=1-p_{b,1}$.}
	\label{fig:OblivFixedAvg_BS}
\end{figure}
\begin{figure}[h!]
	\centering
	\includegraphics[width=0.8\textwidth]{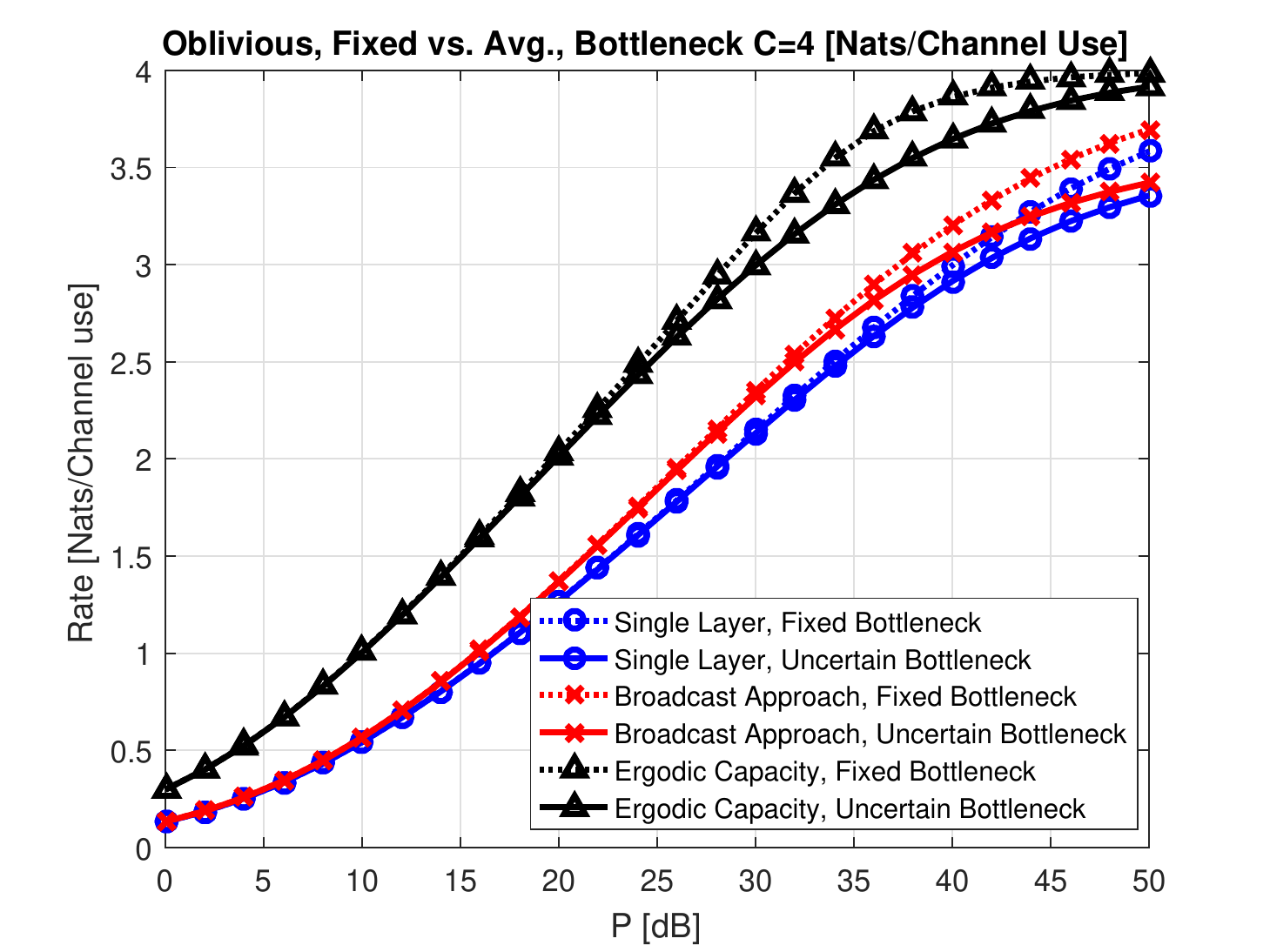}
	\caption{Oblivious comparison of fixed vs. average bottleneck channel capacity ($C=\sum_{i=1}^{2}p_{b,i}C_i = 4$ [Nats/Chan use]). Comparison of single level coding with the broardcast approach and the ergodic capacity.}
	\label{fig:OblivFixedAvg_All}
\end{figure}

\section{Conclusion and Outlook}
This work considers the problem of efficient transmission over the block fading channel with a bottleneck limited channel. Two main approaches are considered, the first is an oblivious approach, where the sampled noisy observations are compressed and transmitted over the bottleneck channel without having any knowledge of the original information codebook. This is compared to a non-oblivious approach where the sampled noisy data is decoded, and whatever is successfully decoded is reliably transmitted over the bottleneck channel. The model is extended for an uncertain bottleneck channel capacity setting, where transmitter is not aware of the available backhaul capacity per transmission, only its distribution. 
In both settings it is possible to analytically describe in closed form expressions, the optimal continuous layering power distribution which maximizes the average achievable rate. Fortunately, it is also possible to define, and solve numerically the joint optimization of the broadcast approach, under backhaul capacity uncertainty. In addition, as the relay can perform successive refinement source coding matched to backhaul capacity distribution, it does not help and cannot increase expected achievable rate if the relay is informed with
the available capacity per codeword.

A possible direction for further research on the information bottleneck channel is to consider a model with two relays, known as the diamond channel. In the oblivious non-fading case the optimal transmission and relay compression, together with joint decompression at the receiver are known and characterized in \cite{sholomoDiamond19}. For the non-oblivious diamond channel only upper bounds \cite{Michaelsholomo19} and achievable rates of the type \cite{Urbanke19} are available. Another possible direction is extending \cite{Karaksik13} to scenarios where the variable backhaul links capacities $\{C_i\}$ are not available at the relay node, but at the destination only. Another possible extension may include adapting the broadcast MIMO approach for the vector bottleneck channel \cite{ShitzSteiner03}, \cite{6875354}. Another interesting extension is the multi-relay setting where each relay is connected to a finite capacity backhaul link, and under this setting the broadcast approach for multi-access channels \cite{Tajer18} becomes very beneficial.

\section{APPENDIX A}\label{Appendix}
This section provides a proof to Proposition \ref{thm:RbsNonObliv}.
\begin{proof}
	The optimal residual power distribution which maximizes the rate has to maximize the average layered transmission rate under a bottleneck channel capacity limitation $C$.
	The optimization problem can be expressed as
	\begin{eqnarray}\label{proof_1}
	R_{bs,DF,avg} = \max\limits_{I(u)} \frac{1}{2} \int\limits_0^{\infty} du
	(1-F_s(u)) \frac{\rho(u)u}{1+I(u)u}, ~~~ s.t. \left( C - \int\limits_0^{\infty} du \frac{\rho(u)u}{1+I(u)u} \right)\geq 0
	\end{eqnarray}
	This constrained variational problem can be expressed with an Euler-Lagrange constrained optimization problem  
	\begin{eqnarray}\label{proof_2}
	\max\limits_{I(u)} \frac{1}{2} \int\limits_0^{\infty} du
	(1-F_s(u)) \frac{\rho(u)u}{1+I(u)u} + \lambda \left( C - \int\limits_0^{\infty} du \frac{\rho(u)u}{1+I(u)u} \right)
	\end{eqnarray}
	where $\lambda\geq 0$ is a scalar Lagrange multiplier. The problem is a standard constrained variational problem with boundary conditions, with a general form representation
	\begin{eqnarray}\label{proof_3}
	\max\limits_{I(u)} \int\limits_0^{\infty} du
	A\left(u, I(u), I'(u)\right)  + \lambda \int\limits_0^{\infty} du B\left(u, I(u), I'(u)\right)
	\end{eqnarray}
	The Euler-Lagrange condition for extremum is \cite{GF63}
	\begin{eqnarray}\label{proof_4}
	A_I - \frac{d}{du}A_{I'} + \lambda \left( B_I - \frac{d}{du}B_{I'} \right) = 0 
	\end{eqnarray}
	where
	\begin{eqnarray}\label{proof_5}
	A\left(u, I(u), I'(u)\right) = \frac{1}{2} 
	(1-F_s(u)) \frac{-I'(u)u}{1+I(u)u}
	\end{eqnarray}
	\begin{eqnarray}\label{proof_6}
	B\left(u, I(u), I'(u)\right) = \frac{I'(u)u}{1+I(u)u} 
	\end{eqnarray}
	by substituting $A(u,I,I')$ and $B(u,I,I')$ in (\ref{proof_2}). The extremum conditions in (\ref{proof_4}) are computed by the paritial derivatives as function of $I$, $I'$ as follows
	\begin{eqnarray}\label{proof_7}
	A_I & = &\frac{(1-F_s(u))u^2I'(u)}{(1+uI(u))^2}\\
	A_{I'} &= &\frac{-u(1-F_s(u))}{(1+uI(u))}\\
	\frac{d}{du}A_{I'} &=& \frac{uf_s(u)(1+uI(u))+ (1-F_s(u))(u^2I'(u)-1)}{(1+uI(u))^2}\\
	B_I & = &\frac{u^2I'(u)}{(1+uI(u))^2}\\
	B_{I'} &= &\frac{-u}{(1+uI(u))}\\
	\frac{d}{du}B_{I'} &=& \frac{u^2I'(u)-1}{(1+uI(u))^2}\label{proof_8}
	\end{eqnarray}
	Next, substitution of expressions (\ref{proof_7})-(\ref{proof_8}) in the extremum condition equation (\ref{proof_4}), and solving for $I_{opt}(u)$ gives
	\begin{eqnarray}\label{proof_9}
	I_{opt}(u) = \frac{1-F_s(u)+\lambda-uf_s(u)}{u^2f_s(u)}
	\end{eqnarray} 
	which requires applying the boundary conditions to get (\ref{prop_2}). The constant $\lambda_{opt}$ is obtained from applying the boundary condition $I(u_1)=0$, to get  (\ref{prop_3}), and the total rate constraint of the bottleneck channel $C$ is applied by 
	\begin{eqnarray}\label{proof_10}
	C = \int\limits_{u_0}^{U_1} du \frac{-uI'_{opt}}{(1+uI_{opt}(u))} 
	\end{eqnarray} 
	where $I_{opt}$ is specified in (\ref{proof_9}). This leads to the result in (\ref{prop_4}), from
	\begin{eqnarray}\label{proof_11}
	\exp(2C) = \frac{u_1^2f_s(u_1)}{ u_0^2f_s(u_0)} 
	\end{eqnarray}	
	and the equation in (\ref{prop_4}) is directly obtained. 
\end{proof}

\section*{Acknowledgment}

\bibliographystyle{IEEEtran}
\bibliography{IEEEabrv,xbib}

\end{document}